\documentclass[11pt]{article}
\usepackage[T1]{fontenc}
\usepackage{amsfonts}
\usepackage{amsmath}
\usepackage{amssymb}
\usepackage{amsthm}
\usepackage{bbm}
\usepackage{bm}
\usepackage{mathrsfs}
\usepackage{verbatim}
\usepackage{setspace}
\usepackage{color}
\usepackage{pdfsync}
\usepackage{enumitem}

\theoremstyle{plain}
\newtheorem{theorem}{Theorem}[section]

\newtheorem{lemma}[theorem]{Lemma}
\newtheorem{corollary}[theorem]{Corollary}

\theoremstyle{definition}
\newtheorem{definition}[theorem]{Definition}
\newtheorem{remark}[theorem]{Remark}

\newtheorem{example}[theorem]{Example}

\theoremstyle{remark}

\newcommand{\N}{\mathbb{N}}

\newcommand{\R}{\mathbb{R}}

\newcommand{\cF}{\mathcal{F}}

\newcommand{\cH}{\mathcal{H}}

\newcommand{\cM}{\mathcal{M}}
\newcommand{\cN}{\mathcal{N}}

\newcommand{\cP}{\mathcal{P}}
\newcommand{\cQ}{\mathcal{Q}}

\newcommand{\fL}{\mathfrak{L}}

\newcommand{\fC}{\mathfrak{C}}

\newcommand{\bL}{\mathbf{L}}

\newcommand{\bC}{\mathbf{C}}

\DeclareMathOperator{\limmed}{lim\, med}

\DeclareMathOperator{\card}{card}

\newcommand{\sint}{\stackrel{\mbox{\tiny$\bullet$}}{}}

\numberwithin{equation}{section}

\usepackage[pdfborder={0 0 0}]{hyperref}
\hypersetup{
  urlcolor = black,
  pdfauthor = {Marcel Nutz},
  pdfkeywords = {Knightian uncertainty; Nondominated model; Superreplication; Martingale measure; Medial limit; Hahn-Banach theorem},
  pdftitle = {Superreplication under Model Uncertainty in Discrete Time},
  pdfsubject = {Superreplication under Model Uncertainty in Discrete Time},
  pdfpagemode = UseNone
}

\begin{document}

\title{\vspace{-.3em}
Superreplication under Model Uncertainty\\in Discrete Time
\date{First version: January 14, 2013. This version: \today.}
\author{
  Marcel Nutz%
  \thanks{
  Department of Mathematics, Columbia University, mnutz@math.columbia.edu. It is the author's pleasure to thank Bruno Bouchard and Freddy Delbaen for stimulating discussions, as well as the anonymous referees and the Associate Editor for their numerous comments. Research supported by NSF Grant DMS-1208985.
  }
 }
}
\maketitle \vspace{-1em}

\begin{abstract}
We study the superreplication of contingent claims under model uncertainty in discrete time. We show that optimal superreplicating strategies exist in a general measure-theoretic setting; moreover, we characterize the minimal superreplication price as the supremum over all continuous linear pricing functionals on a suitable Banach space. The main ingredient is a closedness result for the set of claims which can be superreplicated from zero capital; its proof relies on medial limits.
\end{abstract}

\vspace{.9em}

{\small



\section{Introduction}\label{se:intro}

A classical probabilistic model of a financial market consists of a measurable space $(\Omega,\cF)$ and a probability measure $P$ determining the distribution of stock prices. Such a model includes randomness in the sense that we do not specify a priori which events will take place; however, we do specify the likelihood of any outcome precisely. In standard pricing procedures, a model is chosen from a parametric family and the parameters are determined through a calibration. It is understood that the model is an ad-hoc tool rather than an accurate representation of the market dynamics; for instance, if the calibration is repeated at a subsequent date, it will lead to different parameter values. %
As a result of the estimation error for the parameters and, possibly more importantly, the choice of the parametric family in the first place, there is a substantial model risk in the pricing and hedging of contingent claims.
Model uncertainty, as we understand it, refers to a situation where the distributions are not assumed to be (completely) known a priori. Rather than having a single probabilistic model, we may want to take into account a whole collection $\cP$ of possible models, each model being represented by a probability measure on $(\Omega,\cF)$. For instance, this collection may include Binomial models with different parameters, but also a mix of different ``types'' of models, e.g., with discrete and continuous marginal distributions.

We study the superreplication of contingent claims under model uncertainty. Using the worst-case approach to consistently address this ``Knightian'' uncertainty,
we require the superreplication is to hold simultaneously under all measures $P\in\cP$ (``$\cP$-q.s.''). More precisely, given an adapted process $S$ and a random variable $f$, we are interested in determining the minimal superreplication price
\[
  x_*(f)=\inf\big\{x\in\R:\,\exists\, H \mbox{ such that } x+ H\sint S_T\geq f\;\cP\mbox{-q.s.}\big\};
\]
here $H$ is a trading strategy (defined simultaneously under all $P\in\cP$) and $H\sint S_T$ is the terminal wealth resulting from trading in $S$ at the discrete dates $t=1,\dots,T$ according to $H$. The difference $x_*(f)+x_*(-f)$ between super- and subreplication price gives a bound for a model risk in the pricing of $f$, relative to the chosen collection $\cP$ of models. Moreover, we want to show that an optimal superreplicating strategy exists; i.e., that the infimum is actually attained for some $H$.

Both problems are well understood in the absence of model uncertainty. Indeed, when $\cP$ contains only one probability measure $P$, then an optimal strategy exists and
\[
  x_*(f)= \sup_{Q\in\cM_e(P)} E_Q[f],
\]
where $\cM_e(P)$ denotes the set of all probability measures $Q$, equivalent to $P$, such that $S$ is a $Q$-martingale. This holds under a no-arbitrage condition which by the fundamental theorem of asset pricing is equivalent to $\cM_e(P)$ being nonempty, and embodies a fundamental duality between wealth processes and the linear pricing functionals $\{E_Q[\,\cdot\,],\, Q\in \cM_e(P)\}$. It is well known that the price $x_*(f)$ can be relatively large for practical purposes; however, superreplication is of primal theoretical importance, for instance, in the solution of portfolio optimization problems. We refer to~\cite{DelbaenSchachermayer.06} and the references therein for the classical theory, and in particular to~\cite{FollmerSchied.04} for the discrete-time case.

Our aim is to obtain similar results in the situation of model uncertainty; i.e., when $\cP$ can have many elements. It is certainly reasonable to suppose that each of the possible models $P\in\cP$ is viable in the usual sense and thus admits an equivalent martingale measure for $S$. As the superreplication problem depends only on the nullsets of the given measures, we may replace each $P\in\cP$ by one of its equivalent martingale measures and as a result, we may assume directly that $\cP$ itself consists of martingale measures. We want to include specifically the possibility that $\cP$ is nondominated in the sense that there exists no reference probability measure $P_*$ with respect to which all $P\in\cP$ are absolutely continuous; in fact, we would like to unify, and interpolate between, the classical case where $\cP$ is a singleton and the completely model-free case where no probabilistic assumptions are made and $\cP$ consists of all martingale measures for $S$, possibly mutually singular. Our main result (Theorem~\ref{th:existence}) states that optimal superreplicating strat\-egies exist in a general measure-theoretic setting. Moreover, $x_*(f)$ is described as the supremum over all continuous linear pricing functionals on a certain Banach space.

While a domination assumption is often made for technical convenience in the literature, nondominated models arise naturally in the representation of robust preferences in decision theory, given the axiom of ambiguity aversion. As shown in~\cite{GilboaSchmeidler.89} for the one-period case, the numerical representation then takes the form of an infimum of expected von Neumann--Morgenstern utilities, and the infimum is taken over a set of measures that is not dominated in general. (See also~\cite{EpsteinJi.2013} for a recent continuous-time model in the context of finance.) On the other hand, allowing for the nondominated case can also result in a computational convenience, similarly as a continuous space $\Omega$ is more likely to lead to explicit results than  a discrete space. As an example, take $S$ to be the canonical process on the path space $\Omega=\R^T$ and let $\cP$ be the set of all probabilities under which $S$ is a martingale (here $S$ models the price of a forward contract). In this situation, $x_*(f)$ and even a corresponding strategy can be computed quite explicitly by means of convex analysis. For the purpose of illustration, consider the case where $T=1$ and assume that $f=f(S_1)$. Then, $x_*(f)$ is given simply by the value $f^\sharp(S_0)$ of the concave envelope of $f$ evaluated at the current stock price $S_0$; moreover, like in the familiar Delta hedging, a superreplicating strategy $H$ can be determined by taking the (right, say) derivative of $f^\sharp$ at $S_0$. This construction can be extended to the multi-period case; cf.\ \cite{BeiglbockNutz.14}. On the other hand, one may wonder if an optimal strategy still exists in a case where we have a prior view on the class $\cP$ of possible models and no explicit formula is available, which is the general situation that we propose to study.

To give another example of a nondominated model, let $S$ again be the canonical process on $\Omega=\R^T$ and let $\cP$ be the set of all probabilities $P$ such that $P$-a.s., $S$ is positive and $S_{t+1}/S_t$ is in a given interval $I$ for all $t$; i.e., there is uncertainty about the log-increments of $S$ and only the bound $I$ is given. (A similar setup was used in~\cite{DolinskyNutzSoner.11} as a discrete approximation to the $G$-Brownian motion of~\cite{Peng.07}.) As a more general example, given a process $S$ on some measurable space, we can prescribe any collection $\cN$ of sets and take $\cP$ to be the family of all martingale measures not charging $\cN$. In the special case where $\cN$ is the collection of nullsets of a given reference measure $P_*$, this yields the classical set of absolutely continuous martingale measures, but generically, it yields a nondominated set.

To the best of our knowledge, there are no previous existence results for superreplication under model uncertainty in discrete time, and more generally for price processes $S$ with jumps (except in the case where strategies are constants; cf.\ \cite{Riedel.11}). There are, however, results for continuous processes $S$ with ``volatility uncertainty'' in specific setups; see \cite{BouchardMoreauNutz.12, FernholzKaratzas.11, NeufeldNutz.12, NutzSoner.10, Peng.10, SonerTouziZhang.2010dual, SonerTouziZhang.2010rep, Song.10} and the references therein. All these results have been obtained by control-theoretic techniques which, as far as we have been able to see, cannot be applied in the presence of jumps. A duality result (without existence) for a specific topological setup in discrete time was obtained in~\cite{DeparisMartini.04}, while \cite{DenisMartini.06} gave a comparable result for the continuous case.

A related topic is the so-called model-free pricing introduced by \cite{Dupire.94,Hobson.98}, where superreplication is achieved by trading in the stock $S$ and (statically) in a given set of options. On the dual side, this is related to the set of all martingale measures for $S$ which are compatible with the prices of the given options, and this set is typically nondominated. In fact, the absence of a dominating measure is a consequence of the absence of modeling assumptions, which is precisely the aim of the approach. A survey can be found in~\cite{Hobson.11}; recent results are \cite{AcciaioBeiglbockPenknerSchachermayer.12, BeiglbockHenryLaborderePenkner.11} in discrete time and
\cite{DolinskySoner.12, GalichonHenryLabordereTouzi.11,HenryLabordereOblojSpoidaTouzi.12} in the continuous case. Once more, we are not aware of general existence results in the case with jumps, and it is worth noting that in contrast to the present work, the mentioned papers consider topological setups and contingent claims which are functions of $S$ alone.

It is clear that, like in the classical case, the price $x_*(f)$ will often be relatively large for practical purposes. Nevertheless, the superreplication problem is a key tool in many situations. As an example, let us consider the mentioned problem of semistatic hedging; clearly, the presence of the options as additional hedging instruments will often reduce the superreplication price significantly. More precisely, let $S$ to be the canonical process on $\Omega=\R^T$ and suppose that the marginal law $\mu$ of $S_T$ is given; this corresponds to the availability of all European options $g(S_T)$ as hedging instruments at price $E_\mu[g]$. The construction of an optimal semistatic superreplication requires a passage to the limit of almost-optimal trading strategies $H^n$ for the stock as well as of options $g^n$. After a suitable normalization of $H^n$, one can apply Komlos' lemma under the given measure $\mu$ to find a limit $g$ of convex combinations of $(g^n)$. Then, the main question is to find an optimal hedging portfolio trading in the stock alone, for the modified claim $f-g(S_T)$, and this is precisely the problem under consideration. (More general forms of semistatic hedging will be studied in forthcoming work.) This example also illustrates that it is important to allow for quite general contingent claims, since it may be very difficult to verify regularity assumption for the limit $g$ when no explicit formula is available.

The main mathematical novelty in this paper is a closedness result (Theorem~\ref{th:closedness}) for the cone $\fC$ of contingent claims which can be superreplicated from initial capital $x=0$. A natural space for this result is introduced; namely, we consider the locally convex space $\fL^1$ of measurable functions with the seminorms given by $\{E_P[|\cdot|],\, P\in\cP\}$. Our result states that $\fC$ is sequentially closed in $\fL^1$ (which is not a sequential space in general), and its proof makes use of the so-called medial limits (see Section~\ref{se:medial} for details). To have closedness rather than sequential closedness, we move to a Banach space $\bL^1\subseteq \fL^1$ whose topology, given by the norm $\|f\|_1=\sup_{P\in\cP} E_P[|f|]$, is stronger than the one of $\fL^1$. The main result is then obtained by a Hahn--Banach separation argument resembling the classical theory. The nondominated situation entails several difficulties at the technical level, in particular regarding limits of strategies. Tools like Komlos' lemma that are crucial in the classical theory fail because there is fundamentally no Hilbertian structure in the nondominated case. At the same time, measurability issues arise, due to a lack of separability and also because there is no reference probability under which one could complete the measure space.

There are at least two obvious questions which are not answered in this paper. First, while $x_*(f)$ is described as the supremum over all continuous linear pricing functionals, we do not establish that (or when, rather)
we have the formula $x_*(f)=\sup_{P\in\cP} E_P[f]$. This question will be studied in follow-up work (cf.\ \cite{BouchardNutz.13}). Second, we do not discuss the possible extension of the present results to the case of continuous-time processes with jumps.

The remainder of this paper is organized as follows. Section~\ref{se:medial} states the necessary facts about medial limits; Section~\ref{se:fL1} introduces the space $\fL^1$; Section~\ref{se:seqClosedness} contains the market model and the closedness result; Section~\ref{se:existence} states the main superreplication result, and the concluding Section~\ref{se:counterex} provides a counterexample showing that $\fL^1$ is not sequential.

\section{Medial Limits}\label{se:medial}

In this section, we state some properties of Mokobodzki's medial limit (cf.~\cite{Meyer.73} or \cite[Nos.~X.3.55--57]{DellacherieMeyer.88}), whose use in the framework of model uncertainty was first introduced in~\cite{Nutz.11int}. In a nutshell, a medial limit is a Banach limit that preserves (universal) measurability and commutes with integration.
Medial limits are usually constructed by a transfinite induction that uses the Continuum Hypothesis (the axiom that $\aleph_1=\card \R$).  We recall that the Continuum Hypothesis is independent of ZFC. In fact, it is known that medial limits exist under significantly weaker hypotheses (e.g., Martin's Axiom, which is compatible with the negation of the Continuum Hypothesis), cf.~\cite[538S]{Fremlin.08}, but not under ZFC alone \cite{Larson.09}. In the remainder of the paper, we assume that medial limits exist\footnote{In fact, we see little reason not to follow the advice of Dellacherie and Meyer~\cite{DellacherieMeyer.78} and ``adopt the Continuum Hypothesis with the same standing as the Axiom of Choice.''}. Moreover, since there are then many medial limits, we choose one and denote it by $\limmed$. It works as follows.

If $\{x_n\}_{n\geq1}$ is a bounded sequence of real numbers, $\limmed x_n$ is a number between $\liminf x_n$ and $\limsup x_n$, and if $\{f_n\}_{n\geq1}$ is a uniformly bounded sequence of random variables on a measurable space $(\Omega,\cF)$, $f=\limmed f_n$ is defined via $f(\omega):=\limmed f_n(\omega)$. The first important property of the medial limit is that $f$ is then universally measurable; i.e., measurable with respect to the $\sigma$-field
\[
  \cF^*:= \bigcap_{P\in M_1(\Omega,\cF)} \cF\vee \cN^P,
\]
where $M_1(\Omega,\cF)$ is the set of all probability measures on $\cF$ and $\cN^P$ is the collection of $P$-nullsets. In particular, if $\cF$ is universally complete (i.e., $\cF=\cF^*$), then the medial limit preserves $\cF$-measurability. The second important property is that $\limmed$ commutes with integration; that is,
\[
  \int (\limmed f_n)\,d\mu = \limmed \int f_n \,d\mu
\]
whenever $\mu$ is a finite (possibly signed) measure on $(\Omega,\cF)$. We refer to \cite[Theorem~2]{Meyer.73} for these results. The medial limit can be extended to nonnegative sequences $\{x_n\}_{n\geq1}$ via
\[
  \limmed x_n := \sup_{m\in \N} \limmed (x_n\wedge m) \in [0,\infty].
\]
If $\{x_n\}_{n\geq1}$ is a general sequence, we set
\[
  \limmed x_n :=\limmed x_n^+ - \limmed x_n^-
\]
provided that the limits on the right-hand side are finite (or, equivalently, that $\limmed |x_n|<\infty$).
The following properties are consequences of the fact that $\limmed$ commutes with integration; cf.\ the proofs of \cite[Theorems~3,4]{Meyer.73}.

\begin{lemma}\label{le:Meyer}
  Let $\{f_n\}_{n\geq1}$ be a sequence of random variables on $(\Omega,\cF)$ and set
  \[
    f:=
  \begin{cases}
    \limmed f_n & \text{if } \limmed |f_n|<\infty, \\
    +\infty & \text{otherwise}.
  \end{cases}
  \]
  Moreover, let $\mu$ be a finite signed measure on $(\Omega,\cF)$.
  \begin{enumerate}[topsep=3pt, partopsep=0pt, itemsep=1pt,parsep=2pt]
    \item We have $\int |f|\,d|\mu| \leq \sup_n \int |f_n|\,d|\mu|$.

    \item If $\{f_n\}_{n\geq1}$ is $\mu$-uniformly integrable, then $f$ is $\mu$-integrable and we have
     $\int f\,d\mu = \limmed \int f_n\,d\mu$.

     \item If $\{f_n\}_{n\geq1}$ converges in measure $\mu$ to some $\mu$-a.e.\ finite random variable $g$, then
      $f=g$ $\mu$-a.e.
   \end{enumerate}
\end{lemma}

\section{The Space $\fL^1$}\label{se:fL1}

Let $\cP$ be a collection of probability measures on a measurable space $(\Omega,\cF)$.
A subset $A\subseteq \Omega$ is called \emph{$\cP$-polar} if $A\subseteq A'$ for some $A'\in\cF$ satisfying $P(A')=0$ for all $P\in\cP$ and a property is said to hold \emph{$\cP$-quasi surely} or \emph{$\cP$-q.s.\ }if it holds outside a $\cP$-polar set.
Consider the set of $\cF$-measurable, real-valued functions on $\Omega$ and identify any two functions which coincide $\cP$-q.s. We denote by $\fL^0=\fL^0(\Omega,\cF,\cP)$ the set of all such equivalence classes; in the sequel, we shall often not distinguish between these classes and actual functions.

\begin{definition}\label{def:frakL1}
  The space $\fL^1(\Omega,\cF,\cP)$ consists of all $f\in \fL^0(\Omega,\cF,\cP)$ such that $\|f\|_{L^1(P)}:=E_P[|f|]<\infty$ for all $P\in\cP$. We equip $\fL^1(\Omega,\cF,\cP)$ with the Hausdorff, locally convex vector topology induced by the family of seminorms $\{\|\,\cdot\,\|_{L^1(P)}:\,P\in\cP\}$.
\end{definition}

To wit, a net $\{f_\lambda\}$ in $\fL^1=\fL^1(\Omega,\cF,\cP)$ converges to $f\in\fL^1$ in $\fL^1$ (i.e., in the topology of $\fL^1$) if and only if
$E_P[|f_\lambda-f|]\to0$ for all $P\in\cP$; that is, if convergence holds in each of the spaces $L^1(P)$.
It is important that the closedness of a set in $\fL^1$ is \emph{not determined by sequences} in general (cf.\ Example~\ref{ex:notMazur}); i.e., we have to distinguish sequential closedness and topological closedness. This is at the heart of certain difficulties that we have encountered in our study; for instance, it is the reason why the problem mentioned in Remark~\ref{rk:strongDualityQuestion}(ii) is nontrivial.

The space $\fL^1(\Omega,\cF,\cQ)$ is defined similarly when $\cQ$ is a family of finite, possibly signed measures.
In accordance with the usual notion of boundedness in a topological vector space, we shall say that a subset $\Theta\subseteq \fL^1(\Omega,\cF,\cQ)$ is \emph{bounded} if
\[
  \sup_{f\in \Theta} \|f\|_{L^1(Q)}<\infty\quad\mbox{for all}\quad Q\in\cQ.
\]

The following is easily deduced from Lemma~\ref{le:Meyer}.

\begin{lemma}\label{le:medialBounded}
  Let $\cF$ be universally complete and let $\{f_n\}_{n\geq1}$ be a bounded sequence in $\fL^1(\Omega,\cF,\cQ)$. Then
  \[
    \{\limmed |f_n|=\infty\}\quad \mbox{is $\cQ$-polar}
  \]
  and $f:=\limmed f_n$ defines an element of $\fL^1(\Omega,\cF,\cQ)$ satisfying
  \[
    \|f\|_{L^1(Q)}\leq \sup_n \|f_n\|_{L^1(Q)}\quad \mbox{for all}\quad Q\in\cQ.
  \]
  Moreover, if $\{f_n\}_{n\geq1}$ has a limit $g$ in $\fL^1(\Omega,\cF,\cQ)$, then $f=g$ $\cQ$-q.s.
\end{lemma}

\section{Sequential Closedness of $\fC\subseteq\fL^1$}\label{se:seqClosedness}

Let $(\Omega,\cF)$ be a measurable space equipped with a filtration $(\cF_t)_{t\in\{0,1,\dots,T\}}$, where $T\in \N$.
We shall assume throughout that
\[
  \cF_t\mbox{ is universally complete, for all $t$.}
\]
Moreover, let $S$ be a scalar adapted process, the \emph{stock price process}. We consider a nonempty set $\cP$ of \emph{martingales measures} for $S$; i.e., probability measures under which $S$ is a martingale. We shall denote by $\cH$ the set of predictable processes, the \emph{trading strategies}. Given $H\in\cH$, the corresponding \emph{wealth process} (from initial capital zero) is the discrete-time integral process
\[
  H\sint S=(H\sint S_t)_{t\in\{0,1,\dots,T\}},\quad H\sint S_t=\sum_{u=1}^t H_u \Delta S_u,
\]
where $\Delta S_u=S_u-S_{u-1}$ is the price increment.

The main result of this section is that the cone $\fC$ of all claims which can be superreplicated from initial capital $x=0$ is sequentially closed in $\fL^1=\fL^1(\Omega,\cF,\cP)$. We denote by $\fL^0_+$ the set of ($\cP$-q.s.)\ nonnegative random variables.

\begin{theorem}\label{th:closedness}
  Let $\cP\neq\emptyset$ be a set of martingale measures for $S$ and
  \[
    \fC:=\big(\{H\sint S_T:\, H\in \cH\} - \fL^0_+\big)\cap \fL^1.
  \]
  Then $\fC$ is sequentially closed in $\fL^1$.
\end{theorem}

Before stating the proof of the theorem, we show the following ``compactness'' property; it should be seen as a consequence of the ``absence of arbitrage'' which is implicit in our setup because $\cP$ consists of martingale measures.

\begin{lemma}\label{le:fundLemma1}
  Let $\{W^n= H^n\sint S_T - K^n\}_{n\geq1}\subseteq\fC$ be a sequence which is bounded in $\fL^1$. Then
  for all $t\in\{1,\dots,T\}$,
  \[
    \{H_t^n \Delta S_t\}_{n\geq1}\quad\mbox{is bounded in $\fL^1$.}
  \]
\end{lemma}

\begin{proof}
  It suffices to show that for each $t\in\{1,\dots,T\}$,
  \begin{equation}\label{eq:proofClosedBddnessInL1}
    \{H^n\sint S_t\}_{n\geq1}\quad\mbox{is bounded in $\fL^1$.}
  \end{equation}
  Since $\{W^n\}_{n\geq1}$ is bounded in $\fL^1$ and $K^n$ is nonnegative, it follows immediately that $\{(H^n\sint S_T)^-\}_{n\geq1}$ is also bounded in $\fL^1$. Now fix $n$ and $P\in\cP$ and recall that $P$ is a martingale measure for $S$. Therefore, the stochastic integral $H^n\sint S$ is a local $P$-martingale, but since we already know that $E_P[(H^n\sint S_T)^-]<\infty$, we even have that $H^n\sint S$ is a true martingale; cf.\ \cite[Theorems~1,\,2]{JacodShiryaev.98}. As a result, $E_P[(H^n\sint S_T)^+]=E_P[(H^n\sint S_T)^-]$ for all $n$ and $P$, and therefore, $\{(H^n\sint S_T)^+\}_{n\geq1}$ is bounded in $\fL^1$, like the sequence of negative parts. So far, we have shown that
  \begin{equation}\label{eq:proofClosedBddnessAtT}
    \{H^n\sint S_T\}_{n\geq1}\quad\mbox{is bounded in $\fL^1$.}
  \end{equation}
  To obtain the same statement for $t<T$, we note that for every $P\in\cP$, the martingale property of $H^n\sint S$ yields that
  \[
    \|H^n\sint S_t\|_{L^1(P)} = \|E_P[H^n\sint S_T|\cF_t]\|_{L^1(P)}\leq \|H^n\sint S_T\|_{L^1(P)}
  \]
  since the conditional expectation is a contraction on $L^1(P)$. Hence, \eqref{eq:proofClosedBddnessAtT} implies the claim~\eqref{eq:proofClosedBddnessInL1}.
\end{proof}

\begin{proof}[Proof of Theorem~\ref{th:closedness}.] %
  Let
  $
    W^n= H^n\sint S_T - K^n
  $
  be a sequence in $\fC$ which converges to some $W\in \fL^1$; we need to find $H\in\cH$ such that
  $W-H\sint S_T\leq 0$ $\cP$-q.s. Indeed, being convergent in $\fL^1$, the sequence $\{W^n\}_{n\geq1}$ is necessarily bounded in $\fL^1$; hence, by Lemma~\ref{le:fundLemma1},
  $\{H^n_t\Delta S_t\}_{n\geq1}$ is bounded in $\fL^1$ for fixed $t\in\{1,\dots,T\}$.
  As $S$ is a martingale and in particular integrable under each $P\in\cP$, we can define the \emph{finite} signed measures $Q_{t,P}$ by
  \[
   dQ_{t,P}/dP=\Delta S_t.
  \]
  Let $\cQ=\{Q_{t,P}\}_{P\in\cP}$, then the above means that the sequence $\{H^n_t\}_{n\geq 1}$ is bounded in $\fL^1(\Omega,\cF_t,\cQ)$. Thus, Lemma~\ref{le:medialBounded} implies that $H_t:=\limmed H^n_t$ is finite $\cQ$-q.s. Setting
  $H_t=0$ on the set $\{\limmed |H^n_t|=+\infty\}\in\cF_t$, we obtain a process $H\in\cH$.
  It remains to check that $K:=H\sint S_T -W$ is nonnegative $\cP$-q.s.
  Indeed, since $W^n\to W$ in $\fL^1$, we know from Lemma~\ref{le:medialBounded} that $W=\limmed W^n$ $\cP$-q.s.
  In view of
  \[
    H\sint S_T = \sum_{t=1}^T (\limmed H^n_t) \Delta S_t = \limmed \sum_{t=1}^T H^n_t \Delta S_t = \limmed (H^n\sint S_T)\quad \mbox{$\cP$-q.s.},
  \]
  we conclude that
  \[
    K = H\sint S_T - W = \limmed (H^n\sint S_T -W_n) = \limmed K^n\quad \cP\mbox{-q.s.}
  \]
  As each $K^n$ is nonnegative, the result follows.
\end{proof}

\section{Main Result}\label{se:existence}

In this section, we show that optimal superreplicating strategies exist and we characterize the minimal superreplication price in a dual way.
As in the preceding section, $\cP$ is a nonempty set of martingale measures for the scalar process $S$ which is adapted to the universally complete filtration~$(\cF_t)$.

\begin{definition}
  We introduce the normed vector space
  \[
    \bL^1=\{f\in \fL^1:\, \|f\|_1<\infty\}, \quad \mbox{where} \quad \|f\|_1 = \sup_{P\in\cP} E_P[|f|].
  \]
\end{definition}

We remark that if a nonnegative claim $f\in\fL^0$ can be superreplicated from some finite (deterministic) initial capital, then necessarily $f\in\bL^1$.

\begin{definition}
  We introduce the cone
  \[
    \bC:= \fC\cap \bL^1 \equiv\big(\{H\sint S_T:\, H\in \cH\} - \fL^0_+\big)\cap \bL^1,
  \]
  as well as the set of continuous linear pricing functionals,
  \[
    \Pi=\big\{\ell\in(\bL^1)^*:\, \ell(\bC)\subseteq \R_-\mbox{ and }\ell(1)=1\big\}.
  \]
\end{definition}

Note that $\Pi$ is indeed the set of all continuous and linear pricing mechanisms which are consistent with obvious no-arbitrage considerations. As $\bC$ contains the nonpositive elements of $\bL^1$, we see that $\ell(\bC)\subseteq \R_-$ implies that $\ell$ is positive; i.e., $\ell(f)\geq0$ whenever $f\geq0$ $\cP$-q.s.

It is obvious that $\bL^1\subseteq\fL^1$ and that the topology of $\bL^1$ is stronger than the one of $\fL^1$. Since sequential closedness and topological closedness are equivalent in a normed space (which is indeed the reason for moving from $\fL^1$ to $\bL^1$), the following is then an immediate consequence of Theorem~\ref{th:closedness}.

\begin{corollary}\label{co:closednessWeak}
  Let $\cP\neq\emptyset$ be a set of martingale measures for $S$. Then $\bC$ is closed in $\bL^1$.
\end{corollary}

The following is our main result: an optimal superreplicating strategy exists and the minimal superreplication price is given by the supremum over all linear prices.

\begin{theorem}\label{th:existence}
  Let $\cP\neq\emptyset$ be a set of martingale measures for $S$ and let $f\in \bL^1$. Then
  \begin{equation}\label{eq:dualityWeak}
    \sup_{\ell\in \Pi} \ell(f) = \inf\big\{x\in\R:\,\exists\, H\in \cH\mbox{ such that } x+ H\sint S_T\geq f\;\cP\mbox{-q.s.}\big\}
  \end{equation}
   and the infimum is attained whenever it is not equal to $+\infty$. %
\end{theorem}

We emphasize that a priori, \eqref{eq:dualityWeak} is an identity in $(-\infty,\infty]$, with the usual convention $\inf\emptyset = +\infty$. (As $f\in\bL^1$, the value $-\infty$ is clearly not possible for the left-hand side.)

\begin{proof}
  We first show the inequality ``$\leq$'' in~\eqref{eq:dualityWeak}. Nothing is to be proved if the set on the right-hand side is empty. Hence, let $x\in\R$ and $H\in \cH$ be such that
  \[
    x+ H\sint S_T\geq f\quad\cP\mbox{-q.s.}
  \]
  As $f\in\bL^1$, this implies that
  $(H\sint S_T)^-\leq (f-x)^-\in \bL^1$. On the other hand, as in the proof of Lemma~\ref{le:fundLemma1}, we have $E_P[(H\sint S_T)^+]=E_P[(H\sint S_T)^-]$ for all $P\in\cP$, and so we deduce that $(H\sint S_T)^+\in \bL^1$ as well. As a result, we have that $H\sint S_T\in \bC$. Now let $\ell\in\Pi$; then positivity and the defining properties of $\Pi$ yield that
  \[
    \ell(f)\leq \ell(x+ H\sint S_T)=x+ \ell(H\sint S_T) \leq x,
  \]
  which proves the desired inequality.

  We turn to the inequality ``$\geq$'' in~\eqref{eq:dualityWeak} and the existence of an optimal superreplicating strategy. Let $x:=\sup_{\ell\in \Pi} \ell(f)\in (-\infty,\infty]$. If $x=+\infty$, nothing remains to be shown, so we may assume that $x$ is finite and show that $f\in x+\bC$ (which immediately yields both the inequality and the existence).

  We first consider the case where $f$ is uniformly bounded from above. Suppose for contradiction that $f\notin x+\bC$. Since the convex cone $\bC$ is closed in $\bL^1$ by Corollary~\ref{co:closednessWeak}, the Hahn-Banach theorem yields a continuous functional $\ell: \bL^1\to\R$ such that
  \begin{equation}\label{eq:separationWeak}
    \sup_{W\in\bC} \ell(W) < \ell(f-x)<\infty.
  \end{equation}
  In fact, since $\bC$ is a cone containing zero, $\sup_{W\in\bC} \ell(W)<\infty$ implies that
  \begin{equation}\label{eq:alphaNullWeak}
    \sup_{W\in\bC} \ell(W) =0
  \end{equation}
  and in particular~\eqref{eq:separationWeak} states that
  \begin{equation}\label{eq:alphaNullConseqWeak}
    \sup_{\ell\in \Pi} \ell(f)=x< \ell(f).
  \end{equation}
  Of course, \eqref{eq:alphaNullWeak} shows that $\ell(\bC)\subseteq \R_-$; in particular, $\ell$ is positive. As $f^+$ is bounded, we have $f-x\leq n$ for $n$ large and hence
  \[
    0< \ell(f-x)\leq \limsup_{n\to\infty} \ell(n).
  \]
  This shows that $\ell(1)=n^{-1}\ell(n)>0$. By a normalization, we may assume that $\ell(1)=1$; but then $\ell\in\Pi$, which contradicts~\eqref{eq:alphaNullConseqWeak}.

  It remains to consider the case where $f^+$ may be unbounded. By the above, we have that $(f\wedge n) -x \in \bC$ for all $n\geq0$. Hence, as
  $(f\wedge n) -x \to f-x$ in $\fL^1$, Theorem~\ref{th:closedness} implies that $f-x\in\bC$.
\end{proof}

\begin{remark}\label{rk:strongDualityQuestion}
  (i) It is not hard to see that the theorem is indeed a genuine generalization of the classical superreplication duality mentioned in the Introduction (apart from our assumption that the $\sigma$-fields are universally complete).

  (ii) It is interesting to find conditions guaranteeing that
  \[
    \sup_{\ell\in \Pi} \ell(f) = \sup_{P\in\cP} E_P[f],
  \]
  which, together with~\eqref{eq:dualityWeak}, would yield an even closer analogue of the classical duality. A partial answer (for a special case) can already be found in~\cite{DeparisMartini.04}. The general case is deferred to future work (cf.\ \cite{BouchardNutz.13}).
\end{remark}

\section{A Counterexample}\label{se:counterex}

The subsequent example features a $\sigma$-convex set $\cP$ of martingale measures for the trivial process $S\equiv0$ and shows that a positive, sequentially continuous functional $\ell$ on $\fL^1$ need not be continuous. (Note that $\fC=\fL^1_-$ when $S\equiv0$, so that positivity and $\ell(\fC)\subseteq \R_-$ are equivalent.) In particular, the nullspace of $\ell$ is then a sequentially closed set which is not topologically closed.

\begin{example}\label{ex:notMazur}
  Let $\Omega=[0,1]$, let $\cF$ be its Borel $\sigma$-field and let
  \[
    \cP=\bigg\{ \sum_{k\geq1} \alpha_k \delta_{x_k}:\, \{x_k\}_{k\geq1}\subseteq [0,1],\; 0\leq \alpha_k\leq 1,\; \sum_{k\geq1}\alpha_k=1 \bigg\}.
  \]
  Then the Lebesgue measure $\mu$ induces a sequentially continuous functional on $\fL^1(\Omega,\cF,\cP)$ which is not topologically continuous.
\end{example}

\begin{proof}
  Any $f\in \fL^1$ is bounded, for otherwise there exist $x_n\in[0,1]$ such that
  $|f(x_n)|\geq 2^n$ and therefore $E_P[|f|]=+\infty$ for $P:=\sum_{n\geq1}2^{-n}\delta_{x_n}$, contradicting that $P\in\cP$.
  Moreover, a sequence $f_n$ in $\fL^1$ converges to zero if and only if it is uniformly bounded and converges pointwise; i.e.,
  \begin{equation}\label{eq:exConvExplicit}
    \sup_{n\geq1,\,x\in[0,1]}|f_n(x)|<\infty \quad\quad\mbox{and} \quad\quad  f_n(x)\to0,\quad x\in [0,1].
  \end{equation}
  Indeed,~\eqref{eq:exConvExplicit} implies the convergence in $\fL^1$ by the bounded convergence theorem (applied for each $P\in\cP$). Conversely, let $f_n$ converge to zero in $\fL^1$, then the pointwise convergence must hold since $\delta_x\in\cP$ for all $x\in[0,1]$. Moreover, being convergent in $\fL^1$, $\{f_n\}_{n\geq1}$ must be bounded in $L^1(P)$ for every $P\in\cP$. If $\{f_n\}_{n\geq1}$ is not uniformly bounded, then after passing to a subsequence, there exist $x_n\in[0,1]$ such that
  $|f_n(x_n)|\geq n2^n$. Hence,
  \[
    \sup_{n\geq1} E_P[|f_n|]=+\infty\quad \mbox{for} \quad P:=\sum_{n\geq1}2^{-n}\delta_{x_n},
  \]
  which is again a contraction. Therefore, we have the characterization~\eqref{eq:exConvExplicit} for sequential convergence in $\fL^1$.

  As a consequence, $\ell=E_\mu[\,\cdot\,]$ is a sequentially continuous linear functional on $\fL^1$ for any probability measure $\mu$.
  However, when $\mu$ is the Lebesgue measure, $\ell$ cannot be topologically continuous because otherwise $E_\mu[\,\cdot\,]$ would have to be dominated by finitely many of the seminorms $\{E_P[|\,\cdot\,|],\,P\in\cP\}$, which is clearly not the case.
\end{proof}

\newcommand{\dummy}[1]{}

\end{document}